\documentclass[runningheads,envcountsect,envcountsame]{llncs}

\usepackage{hyperref}
\hypersetup{
    colorlinks=true,
    citecolor=blue,
    }

\usepackage{stix}
\usepackage{mathtools}
\usepackage{stmaryrd}

\usepackage{cleveref} 
\newtheorem{fact}[theorem]{Fact}
\usepackage{appendix}

\newcommand{\dc}{\mathop{\downarrow}\nolimits}
\newcommand{\upc}{\mathop{\uparrow}\nolimits}
\newcommand{\subword}{\preccurlyeq}
\newcommand{\equivdef}{\stackrel{\mbox{\begin{scriptsize}def\end{scriptsize}}}{\Leftrightarrow}}
\newcommand{\step}[1]{\xrightarrow{\!\!#1\!\!}}

\crefname{conjecture}{conjecture}{conjectures}
\Crefname{conjecture}{Conjecture}{Conjectures}

\title{On the State Complexity of Square Root and Substitution for Subword-Closed Languages}
\titlerunning{State Complexity of Square Root and Substitution}

\author{Jérôme Guyot\inst{1}}
\authorrunning{J. Guyot}
\institute{DER Informatique, Univ. Paris-Saclay, ENS Paris-Saclay, Gif-sur-Yvette, France\\ \email{jerome.guyot@ens-paris-saclay.fr}}
\date{}

\begin{document}

\maketitle

\begin{abstract}

This paper investigates the state complexities of subword-closed and superword-closed languages, comparing them to regular languages. We focus on the square root operator and the substitution operator. We establish an exponential lower bound for the $n$-th roots of superword-closed languages. For subword-closed languages we analyze in detail a specific instance of the square root problem for which a quadratic complexity is proven. For the substitution operator, we show an exponential lower bound for the general substitution. In the case of singular substitution, we show a quadratic upper bound when the languages are subword closed and based on disjoint alphabets. We conjecture a quadratic upper bound in the general case of singular substitution with subword-closed languages and prove it in the case where the substitution applies on a directed language.

\end{abstract} 

\newpage
\section*{Introduction}

\subsubsection*{State complexity.}

The number of states of the canonical automaton recognizing a regular language
$L$ is known as its state complexity, denoted $\kappa(L)$. It is a common
measure of the complexity of regular languages. Finite state automaton are often
used as data structure: the size of the automata thus becomes an important
parameter in the  complexity analysis of some algorithms. \\ 

For an operation or a function $f$ on regular languages,  the natural question
is  ``\emph{how does $\kappa(f(L))$ relate to $\kappa(L)$?}'' This leads to the
definition of \emph{the state complexity of $f$} as the function
$\phi_f:\mathbb{N} \to \mathbb{N}$ such that $\phi_f(n)$ is the maximum state complexity of
$f(L)$ with $L$ having state complexity at most $n$. This notion can be extended
to functions having multiple arguments, for example the state complexity of
intersection would be given by $\phi_\cap(n_1,n_2)$. This area of automata theory already has a rich literature and the recent survey
\cite{DBLP:journals/corr/GaoMRY15} describes the known results for a wide range
of operations and classes of languages. As it can be difficult to find the exact
complexity of some $f(L)$ or to give a uniform formula for the function $\phi_f$, the goal is often to obtain bounds on the complexity of $f(L)$ and on $\phi_f$. This induces a classification of the operations on regular languages and finite automata based on the growth of $\phi_f$. 

\subsubsection*{State complexity of subregular classes.}

It is often interesting to measure the state complexity of a function $f$ when
we restrict its argument to a subregular class. As some applications only focus
on a subclass of automata it becomes natural to study state
complexity on this restricted domain. For example, computational
linguistics use automata to encode lexicons that are always finite languages:
they are considered in  \cite{state_comp_finite_lang} while their
complement, cofinite languages, are considered in~\cite{bassino2010}. 
Linguistics are also interested in locally-testable languages~\cite{heinz2009},
and other areas like genomics or databases or pattern matching have their own
subclasses of interest.

The study of state complexity restricted to such subclasses has recently become quite
active after Brzozowski et al.\ initiated a systematic study of state complexity on
various fundamental classes of subregular
languages~\cite{brzozowski2012c,brzozowski2014,quotient_closed,ideals}.

\subsubsection*{Subword-closed and superword-closed languages.}

In the context of computer-aided verification, several algorithms for program
verification use well-quasi-ordered data domains~\cite{finkel,Abdulla2000AlgorithmicAO,Bertrand2013} and in particular, using
Higman's lemma, words ordered by the subword, or subsequence, ordering.\footnote{
        Here the terminology is not fixed. Some authors use ``subword'' for a factor, and use ``scattered
        subword'' for what we call a subword
        (we follow~\cite{sakarovitch83}).
        Superword-closed languages
        are sometimes called ``shuffle ideals'' (\cite{heam}) and even 
        ``all-sided ideals'' (\cite{ideals}). In~\cite{wqo} ``ideals'' 
        denote the directed subword-closed languages, while superword-closed 
        languages are called ``filters''. 
}
These algorithms handle subword-
and superword-closed languages.  

The state complexity of subword- and superword-closed languages has
not been analyzed extensively.\footnote{
For these languages the most studied question is to obtain them by taking subword- or superword-closure of arbitrary regular languages~\cite{heam,gruber2007b,gruber2009,scattered,philippe} and even of nonregular languages~\cite{bachmeier}. 
}
The main results are due to Brzozowski et al.\  who considered
subword-closed languages in~\cite{quotient_closed} and superword-closed
languages in~\cite{ideals}:  these works actually consider several subregular classes at once and only
focus on  the main operations:
 boolean combinations (with $\oplus$ denoting the symmetric difference), concatenation, iteration and mirror (denoted $L^R$). 
However, there exist other interesting operations to
consider as they also preserve the subword/superword closedness such as the shuffle. 

The following tables give the known bounds on the state complexity of a few operations. 

\begin{center}
State complexity for subword-closed languages
\begin{tabular}{|c|c|c|c|}
\hline
Operation & Upper Bound & Tightness requirement & References \\
\hline
$L \cap K$ &  $mn - (m+n-2)$ & $|\Sigma| \geq 2$ & \cite{quotient_closed} Theorem 2 \\

$L \cup K$ & $mn$ & $|\Sigma| \geq 4$ & \cite{quotient_closed} Theorem 2 \\

$L \setminus K$  & $mn-(n-1)$ & $|\Sigma| \geq 4$ & \cite{quotient_closed} Theorem 2 \\

$L \oplus K$ & $mn$ & $|\Sigma| \geq 2$ & \cite{quotient_closed} Theorem 2 \\

$L \cdot K$ & $m+n-1$ & $|\Sigma| \geq 2$ & \cite{quotient_closed} Theorem 3 \\

$L^*$ (and $L^+$) & 2 & $|\Sigma| \geq 2$ & \cite{quotient_closed} Theorem 4 \\

$L^R$ & $2^{n-2}+1$ & $|\Sigma| \geq 2n$ & \cite{quotient_closed} Theorem 5 \\

$L^k$ & $k(n-1) + 1$ & $|\Sigma| \geq 2 $ & \cite{hospod} Theorem 9 \\

\hline

\end{tabular}
\end{center}

\hfill \\

\begin{center}
State complexity for superword-closed languages
\begin{tabular}{|c|c|c|c|}
\hline
Operation & Upper Bound & Tightness requirement & References \\
\hline
$L \cap K$ & $mn$ &  $|\Sigma| \geq 2 $ & \cite{ideals} Theorem 7 \\

$L \cup K$ & $mn - (m+n-2)$ &  $|\Sigma| \geq 2 $ & \cite{JALC-2010-071} Theorem 5 \\

$L \setminus K$ & $mn-(m-1)$ &  $|\Sigma| \geq 2 $ &\cite{JALC-2010-071} Theorem 5 \\

$L \oplus K$ & $mn$ &  $|\Sigma| \geq 2 $ & \cite{ideals} Theorem 7 \\

$L \cdot K$ & $m+n-1$ &  $|\Sigma| \geq 1 $ & \cite{ideals} Theorem 9 \\

$L^*$ (and  $L^+$) & $n+1$ &  $|\Sigma| \geq 2 $ & \cite{ideals} Theorem 10 \\

$L^R$ & $2^{n-2}+1$ & $|\Sigma| \geq 2n - 4$ & \cite{ideals} Theorem 11 \\

$L^k$ & $k(n-1) + 1$  & $|\Sigma| \geq 1 $ & our \Cref{borne_L_k} \\

\hline

\end{tabular}
\end{center}
\hfill \\

\subsubsection*{Our contribution.}

We are interested in completing the picture and consider the state complexity of
other operations on subword-closed or superword-closed languages.

In the following sections, we focus on two operators: the
$n^{th}$-root operators and the substitution operator. For the root operators,
we show that they have  exponential complexity even when restricted to
 superword-closed languages. For the
subword closed languages, when $L = \dc (\{ w \})$ is the downward closure of a word $w$ it seems that
there is a quadratic upper bound, we do not know if this extends to the general
case of subword-closed languages. For the substitution operator $L^{a\leftarrow K}$, we show a
quadratic upper bound when $L$ and $K$ are subword-closed and based on disjoint
alphabets and conjecture a quadratic upper bound when $L$ and $K$ are subword-closed. Finally we proved a quadratic upper bound in the case where $L$ is directed  (without any hypothesis on the alphabets). \\

This work contributes to understanding the state complexities of subregular languages. It was done in the context of an initiation to research project at ENS Paris-Saclay. I warmly thank Philippe Schnoebelen for his valuable help and dedication to the project. I also thank Maelle Gautrin and Simon Corbard for initiating the research on the substitution operator.

\section{Preliminaries}

In this work we assume that the reader is familiar with finite automaton and regular languages. We assume a finite alphabet $\Sigma=\{a,b,\ldots\}$, use $\epsilon$ to denote the empty word, and write concatenation either $u \cdot v$ or $uv$. We write $\Sigma(u)$ for the set of letters occurring in a word $u$, and $|u|_a$ for the number of occurrences of letter $a$ in $u$.\\  

As we are going to work with subword closed languages that we call here downward closed languages, we must define the notion of subwords. 

\begin{definition}\label{subword_def}

Let $x,y \in \Sigma^*$, $x$ is a subword of $y$, denoted $x \subword y$, if and only if $y$ can be written as $y = u_1 x_1 \cdots u_n x_n u_{n+1}$
for some factorization  $x = x_1 \cdots x_n$ of $x$ and.
  some words $u_0,\ldots,u_n$, $u_i \in \Sigma^*$.

\end{definition}

\begin{example}

$\mathtt{abba}$ is a subword of $\mathtt{accbebda}$. 

\end{example}

Once we have the notion of subwords, we can take the subwords of a language $L$, this is called the downward closure of $L$. 

\begin{definition}\label{down_closure_def}

Let $L \subseteq \Sigma^*$, $\dc (L) = \{ x \in \Sigma^* \ | \ \exists \ y \in L, \ x \subword y \}$ is the \emph{downward closure} of $L$. $L$ is said downward closed if and only if $\dc (L) = L$.

\end{definition}

\begin{example}

$ \dc(\mathtt{a}^m)= \{ \mathtt{a}^i \ | \ 0 \leq i \leq m \} $ and  $ \{ \epsilon, \mathtt{a}, \mathtt{b}, \mathtt{ab} \} = \dc (ab) $ are downward closed.

\end{example}

There is a symmetric notion to downward closed languages. The idea is not to take the subwords of words of $L$ but to take the words such that the words of $L$ are subwords of them : superwords of $L$. This gives the upward closure of $L$. 

\begin{definition}\label{up_closure_def}

Let $L \subseteq \Sigma^*$, $\upc (L) = \{ x \in \Sigma^* \ | \ \exists \ y \in L, \ y \subword x \}$ denotes the \emph{upward closure} of $L$. $L$ is said upward closed if and only if $\upc (L) = L$. 

\end{definition}

\begin{example}

With $\Sigma = \{\mathtt{a}\}$, $ \upc(\mathtt{a}^m)= \mathtt{a}^m\mathtt{a}^* $ and  $\mathtt{a}^m\mathtt{a}^*$ is upward closed.

\end{example}

\begin{remark}

When taking the upward closure we need to specify the alphabet we are using. This is seen in the equality $\upc(L) = L \shuffle \Sigma^*$ linking superwords and shuffle.\footnote{
We do not recall the definition of the shuffle, denoted $L\shuffle L'$, of  languages. It can be found, e.g., in \cite{ideals} or \cite{heam}.
}
The alphabet will always be clear from the context when we talk of upward closed languages or upward closures.

\end{remark}

\begin{remark}

The downward and upward closures verify the Kuratowski closure axioms :

Let $L,K \subseteq \Sigma^* $ 

\begin{itemize}
\item $\dc (\emptyset) = \emptyset $ and $\upc (\emptyset) = \emptyset $.
\item $L \subseteq \dc (L)$ and $L \subseteq \upc (L)$.
\item $\ \dc (\dc (L)) = \dc (L) $ and $\upc (\upc (L)) = \upc (L)$. 
\item $\dc(L \cup K) = \dc(L) \cup \dc(K)$ and $\upc(L \cup K) = \upc(L) \cup \upc(K)$.

\end{itemize}

\end{remark}

The reason why we present the downward closedness and the upward closedness as dual notions comes from the following fact :

\begin{fact}
$L$ is downward closed if and only if its complement $\Bar{L}$, i.e., $\Sigma^*\setminus L$, is upward closed. Equivalently, $L$ is upward closed if and only if $\Bar{L}$ is downward closed.
\end{fact}

Another important fact is Haines Theorem: if $L\subseteq \Sigma^*$ is downward closed or upward closed then it is regular.

Let us now define  the state complexity of a regular language. 

\begin{definition}\label{state_complexity_def}

Let $L \subseteq \Sigma^*$ be a regular language. We write $sc(L)$ for  the number of states of the canonical automaton recognizing $L$ and call it
 the \emph{state complexity} of $L$.

\end{definition}

\begin{remark}

The minimal DFA is just the canonical DFA with unproductive transitions and states removed.

\end{remark}

However, reasoning on automaton is not always the easiest way to write formal proofs. For this we will often prefer using the formalism of left-quotients that we will just call \emph{quotients}. For a language $L$ we let $\mathcal{R}(L)$ denote the set of quotients of $L$ and we write
$\kappa(L)$ for   $|\mathcal{R}(L)|$, i.e., the number of different quotients of $L$.

\begin{theorem}[\cite{JALC-2010-071}]
For $L$ a regular language, $sc(L)=\kappa(L)$.
\end{theorem}

\begin{remark}
Historically, this theorem was proven by Nerode and Myhill in 1957.

The notation $\mathcal{R}(L)$ should be read as $\mathcal{R}_{\Sigma}(L)$ since the set of quotients depends on the alphabet. Again we leave the alphabet implicit when there is no ambiguity. 

\end{remark}

\begin{example}

$ \kappa(\{\mathtt{a}^m\})= m+2 $ as the quotients are $\{ \mathtt{a}^i \ | \ 0 \leq i \leq m \}$ and $\emptyset$.

\end{example}

Let us introduce the notion of state complexity through a less trivial example. For this we prefer the notation $L/x$ (over  $x^{-1}L$) for a left quotient. This operation has a counterpart in the derivatives of regular expressions, which we  denote similarly with $e/x$. 
 
\begin{lemma}\label{state_complexity_word}
    For any word  $w$, one has  $\kappa(\dc (w)) = |w| + 2$.
\end{lemma} 

\begin{proof}
We show that    the quotients of $\dc (w)$ are exactly the empty set and the  $\dc (v)$ for all $v$ that are suffix of $w$, thus $\kappa(\dc (w)) = |w| + 2$. \\
 
    Let $x \in \Sigma^*$, if $x$ is not a subword of $w$ then $\dc (w) / x = \emptyset$. If $x \subword w$, let $w_1$ be the smallest prefix of $w$ such that $x \subword w_1$ and write $w = w_1 w_2$, then $\dc (w) /x = \dc (w_2) $. In fact, for any subword $y$ of $w_2$, $xy \subword w$. And for $y$ such that $xy \subword w$, then we can factorize $w$ in $w = w_1 w_2$ such that $x \subword w_1$ and $y \subword w_2$. 
\qed
\end{proof}

To understand better how we compute quotients of downward closed languages we introduce the following lemma. This is an essential property that will be often used when studying the structure of quotients. 

\begin{lemma}

Let $L$ be a downward closed language and $x,y$ two words with $x \subword y$. Then $L/y \subseteq L/x$. 	

\end{lemma}

\begin{proof}
Let $w \in L/y$, then we get by definition $yw \in L$. However, as $x \subword y$, we also have $xw \subword yw$. As $L$ is downward closed, this implies $xw \in L$, hence $w \in L/x$. 
\qed
\end{proof}

\begin{remark}
$\emptyset$ is a quotient of $L$ if and only if $L \neq \Sigma^*$. 

\end{remark}

When it is too hard to exactly find the set of quotients, we can use dividing sets to obtain lower bounds on the state complexity of the language. 

\begin{definition}
	Let $L$ a language and $\mathcal{F}$ a set of words, $\mathcal{F}$ is a dividing set for $L$ iff $\forall \ x\neq y \in \mathcal{F}, L/x \neq L/y $. 
\end{definition}

\begin{remark}
	By definition of a dividing set we have that if $\mathcal{F}$ is a dividing set of $L$ then $\kappa(L) \geq |\mathcal{F}|$. 
\end{remark}

\section{Root operators}

In this section we will be working on the root operators, that are defined as the $n^{th}$-root and the union of those roots. Formally:

\begin{definition} Let $L \subset \Sigma^*$, and $k \in \mathbb{N}_{>0}$, we let
\begin{xalignat*}{2}
\sqrt[k]{L} &= \{ x \ | \ x^k \in L \}\;,
&
\sqrt[*]{L} &= \bigcup_{k \in \mathbb{N}_{>0}} \sqrt[k]{L}\:.
\end{xalignat*}
When $k=2$ we may just write $\sqrt{L}$.
\end{definition}

As recalled in the introduction, the root operators preserve  regularity and  downward/upward closedness. It is known that the state complexity of these operations are exponential in the general case of regular languages~\cite{krawetz2005}. We will prove exponential lower bounds in the case of upward closed languages. For downward closed languages, we will only focus on the square root and study the case of $L = \dc (w)$ for $w$ a word to conjecture a quadratic upper bound.

\subsection{Upward closed languages}\label{sqrt_upward_close}

Let us first show that the root operator preserves  upward closedness. 

\begin{proposition}

	Let $L$ be a upward closed language, then $\sqrt[k]{L}$ and $\sqrt[*]{L}$ are upward closed for $k \in \mathbb{N}_{>0}$. 

\end{proposition}

\begin{proof}
	Let $k \in \mathbb{N}$, let $x \subword y \in \sqrt[k]{L}$. Then $x^k \in L $ and $x^k \subword y^k $. As $L$ is upward closed, $y^k \in L$. Thus $y \in \sqrt[k]{L}$.
	
	If $x \subword y \in \sqrt[*]{L}$, then there exists $k \in \mathbb{N}$ such that $x \subword y \in \sqrt[k]{L}$ and we conclude as before.
\qed
\end{proof}

Let $\Sigma_n$ be an alphabet having at least $n$ distinct letters $a_1,... \ ,a_n$. We denote by $V_n$ the $n$ letter word of length $n$ : $ V_n = \mathtt{a}_1 \cdots \mathtt{a}_n $. For example, $V_4 = \mathtt{abcd}$.

\begin{proposition}\label{prop-2-n-quotients}
    $\kappa(\sqrt[k]{\upc (V_n)}) \geq 2^n$ when $k\in \mathbb{N}_{>0}$ and $\kappa(\sqrt[*]{\upc (V_n)}) \geq 2^n$.
    
\end{proposition}

\begin{proof}

    Let $v,w$ be two subwords of $V_n$ such that $v \neq w$. 
    Then there is a letter of $V_n$ that is either in $v$ and not in $w$ or in $w$ and not in $v$. Without loss of generality, let us assume the letter $a$ is in $v$ but not in $w$.

    Then let $x$ be $V_n$ with $a$ removed. 
    Thus, for any $k \geq 1$, $V_n$ is a subword of $(vx)^k$ but not of $(wx)^k$. Thus, $vx \in \sqrt[k]{\upc (V_n)}$ but $wx \notin \sqrt[k]{\upc (V_n)}$. Hence, $\sqrt[k]{\upc (V_n)}/v \neq \sqrt[k]{\upc (V_n)}/w$ since one contains $x$ but not the other. Finally, as there are $2^n$ distinct subwords of $V_n$, there has to be at least $2^n$ distinct quotients in $\mathcal{R}(\sqrt[k]{\upc (V_n)})$. \\
    
    Now for $\kappa(\sqrt[*]{\upc (V_n)})$, as $wx \notin \sqrt[k]{\upc (V_n)}$ for all $k \in \mathbb{N}_{>0}$, $\mathcal{R}(\sqrt[*]{\upc (V_n)})$ contains at least $2^n$ quotients. 
\qed    
\end{proof}

This shows a lower bound in $2^n$, however experimentally for the square root, we have values close to $3^{n-1}$. In the appendix, we improve \Cref{prop-2-n-quotients} and give in \Cref{better_dividing_set_sqrt_down} a lower bound in $\mathcal{O}( 2.41^n )$. Furthermore, we give in \Cref{cut_shuffle} a direct characterization of $\sqrt{\upc(V_n)}$ by describing its minimal elements (i.e. its generators).

\subsection{Downward closed languages}

Let us first show that the root operator preserves downward closedness. 

\begin{proposition}

	Let $L$ be a downward closed language, then $ \forall k \in \mathbb{N}_{>0}, \ \sqrt[k]{L}$ and $\sqrt[*]{L}$ are downward closed.

\end{proposition}

\begin{proof}

	Let $k \in \mathbb{N}_{>0}$, let $y \subword x \in \sqrt[k]{L}$. Then $y^k \subword x^k \in L$. As $L$ is downward closed, $y^k \in L$. Thus $y \in \sqrt[k]{L}$. \\
	
	If $y \subword x \in \sqrt[*]{L}$, then there exists $k \in \mathbb{N}$ such that $y \subword x \in \sqrt[k]{L}$ and we conclude as before. 
\qed
\end{proof}

Let $V_n=a_1a_2\cdots a_n$ be like in \Cref{sqrt_upward_close} and define $W_n=V_n^3$. For example $W_2=\mathtt{ababab}$. To understand better what could be the square root of a downward closure, we study the example of $\sqrt{\dc(W_n)}$. \\

Recall that the conjugates of $u$ are all the words of the form $v'v$ for some factorization $vv'$ of $u$. E.g., the conjugates of $\mathtt{babar}$ are $\{\mathtt{babar,abarb,barba,arbab,rbaba\}}$.

\begin{proposition}\label{subword_conjugate}
$x\in \sqrt{\dc (W_n)}$ iff $x$ is a subword of a conjugate of $V_n$. 
\end{proposition}

\begin{proof}
    Let us proceed by double implications, \\ 

    Let us start with the easy case, let $x$ be a subword of a conjugate $v'v$ of $V_n = vv'$. Then $xx$ is a subword of $v'vv'v$ which is a subword of $vv'vv'vv' = W_n$. Thus $x \in \sqrt{\dc(W_n)}$. \newline 

    Now for the $\implies$ direction, let $x \in \sqrt{\dc(W_n)}$, if $x = \epsilon$ then it is a subword of any conjugate. Otherwise, let us isolate the first letter of $x = ay$ and factorize $V_n = vav'$ with no $a$ occurring in $v$: we get $xx = ayay \subword vav'vav'vav'$. There are two cases : either the first $x$ embeds in the $vav'v$ prefix, or the
second $x$ embeds in the $v'vav$ suffix. In the 1st case $x \subword vav'v$ entails $x \subword av'v$ since $v$ has no $a$ and we are done since $av'v$ is a conjugate of $V_n$. In the second case $x\subword v'vav$ entails $x \subword av$ since $v'v$ has no $a$ and we are done again.    
\qed
\end{proof}

\begin{proposition}
    $\kappa(\sqrt{\dc (W_n)})= n^2 -n+3 $.  
\end{proposition}

\begin{proof}

	Using \Cref{subword_conjugate}, we have that $x\in \sqrt{\dc (W_n)}$ iff $x$ is a subword of a conjugate of $V_n$. Let us compute $\kappa(\sqrt{\dc (W_n)})$ by listing and counting its quotients. \\
	
	Let us write $L = \sqrt{\dc (W_n)}$. First, if $x = \epsilon$, $L/x = L$. If $x$ contains two times the letter $a$, then $L/x = \emptyset$. And if $x$ is a conjugate of $V_n$, $L/x = \epsilon$.\\
	
	There only remain the cases where $x$ is a strict non-empty
subword of a conjugate of $V_n$. We claim that all these $x$ yield different $L/x$ quotients. Let us write $x = a_{i_1}...a_{i_m}$ with $1 \leq m \leq n$. Let us denote by $v[i,j]$ the factor of $V_n$ starting by $a_i$ and finishing by $a_j$, $1 \leq i\leq j \leq n$. We get that the conjugates of $V_n$ are of the form $v[i,n] v[1,i-1]$. Thus, by \Cref{subword_conjugate} we have $L = \sum_{1 \leq i \leq n} \dc(v[i,n] v[1,i-1]) $. Then we get that if $i_1 < i_m < n$, 
\begin{align*}
	L/x &= \bigcup_{1 \leq i \leq n} \dc(v_i^n v_1^{i-1})/(a_{i_1}...a_{i_m}) \\
		&= \bigcup_{1 \leq i \leq i_1} \dc( v_{i_{m+1}}^n v_1^{i-1} )\\
		&= \dc(v_{i_{m+1}}^n v_1^{i_1-1})
\end{align*}
	Furthermore,  using the same equations, if $1 \leq i_m \leq i_1 -1$, $L/x = \dc v_{i_{m+1}}^{i_1-1}$. Thus we get a different quotient, different from $L$, $\{ \epsilon \}$ or $\emptyset$, for each pair $(i_1,i_m) \in \{1,...,n\} \times \{ i+1 [n],...,i+(n-1) [n] \}$ where $x[n]$ denotes $x$ mod $n$. Thus, we end up with $n^2 -n +3$ different quotients.
\qed  
\end{proof}	

We implemented an algorithm that given a language returns the minimal DFA of the square root of its downward closure. In practice as it can be quite slow we used it on words of length smaller than $12$. Based on the results, it seems that $W_n$ gives the largest state complexity for words of length $3n$. In the following, we try to understand this by studying the largest state complexity of $\sqrt{\dc (w)}$ for $w$ word of length $n$.

\begin{definition}
    Let $\alpha(n) = \max_{|w| \leq n} \kappa(\sqrt{\dc(w)})$. 
\end{definition}

\begin{proposition}
    For all $n \in \mathbb{N}$ we have 
    \begin{itemize}
        \item   $\alpha(n) \leq \alpha(n+1)$,
        \item   $\alpha(n) < \alpha(n+2)$.
    \end{itemize}
  
\end{proposition}

\begin{proof}
	For the first point, let $w$ be a word reaching the maximum for $n$ : $\kappa(\sqrt{\dc(w)}) = \alpha(n)$. 
    Let us consider $wa_{n+1}$, $\sqrt{\dc(wa_{n+1})} = \sqrt{\dc(w)}$. This comes from the fact that a letter appearing once does not appear in the square root ! 
    
    Thus $\kappa(\sqrt{\dc(w)}) = \kappa(\sqrt{\dc(wa_{n+1})}) \leq \alpha(n+1)$. \\ 

    For the second point, let $w$ be a word reaching the maximum for $n$. Let $u$ a word of maximal length in $\sqrt{\dc(w)}$. This means that $u$ is a maximal word such that $uu$ is a subword of $w$. Thus we can factorize $w$ as $w=w_1 w_2$ where $u \subword w_1$ and $u \subword w_2$. Let us now consider the word $v = w_1a_{n+1}w_2a_{n+1}$. We can easily see that $ua_{n+1}$ is a maximal word of $\sqrt{\dc v}$, and that $\sqrt{\dc(w)} \subsetneq \sqrt{\dc(v)}$. \\
    
    It is important to see that a dividing set of $\sqrt{\dc(w)}$ is still a dividing set of $\sqrt{\dc(v)}$. To show this, let us take $x,y$ such that $\sqrt{\dc(w)}/x \neq \sqrt{\dc(w)}/y $. Then there exists a word $z$ made of letter $\{ a_1, ..., a_n\}$ such that $xz \in \sqrt{\dc(w)}$ but $yz \notin \sqrt{\dc(w)}$ (or the other way around, but this case is symmetric and leads to the same conclusion). Seeing that all the new words in  $\sqrt{\dc(v)}$ are of the form $ua_{n+1}$ for $u \in \sqrt{\dc(w)}$, we get that $yz \notin \sqrt{\dc(v)}$. This concludes the validity of the dividing set.\\
    
    Now we can also add a new element to this dividing set of $\sqrt{\dc(v)}$ : $ua_{n+1}$. To see this, we can just observe that if $\sqrt{\dc(w)}/x = \epsilon$ then $\sqrt{\dc(v)}/x = \{ \epsilon, a_{n+1} \} $ and $\sqrt{\dc(v)}/ua_{n+1} = \{ \epsilon \}$. Indeed, if $x \in \sqrt{\dc(w)}$, then $xa_{n+1}xa_{n+1} \subword w_1a_{n+1}w_2a_{n+1} = v$, and the second equality comes from the paragraph above. Thus we can add $ua_{n+1}$ to the dividing set, thus giving that $\alpha(n+2) \geq \alpha(n) +1 > \alpha(n)$. 
\qed    
\end{proof}

\begin{corollary}
    Let $w$ be a word of length $n+2$, if $w$ has $2$ letters appearing only once, then $\kappa(\sqrt{\dc(w)}) < \alpha(n+2)$. 
\end{corollary}

\begin{proof}
	Let $a,b$ be the two letters appearing only once in $w$ and $w'$ be the word $w$ without $a$ and $b$, we get $|w'| = n$. Let $x \in \sqrt{\dc(w)}$ then $xx \subword w$. However, as there is only one $a$ and $b$ in $w$, $x$ cannot contain an $a$ or $b$. Thus, $x \subword w \Leftrightarrow x \subword w'$ and $\kappa(\sqrt{\dc(w)}) = \kappa(\sqrt{\dc(w')}) \leq \alpha(n) < \alpha(n+2)$.
\qed    
\end{proof}

Let us introduce the notation $c(n) = \lceil \frac{n}{3} \rceil $. Let $U_n$ be defined as $U_n = V_{c(n)} V_{c(n)} V_{n-2c(n)}$. For example, $U_7 = \mathtt{abcabca}$ and $U_{11}= \mathtt{abcdabcdabc}$. \\

\begin{conjecture}
    
    $\forall \ n \in \mathbb{N}_{>0}, \ \kappa(\sqrt{\dc(U_n)}) = \alpha(n) \leq c(n)^2 -c(n) +3 $.

\end{conjecture}

\begin{remark}

	Experimentally, we observed that for $n\leq 11$, $\alpha(n) \leq c(n)^2 -c(n) +3$. The question is now to see if this can extend to larger words and be proven.

\end{remark}

	\section{Substitution operator}

Let us now consider the substitution operator. We consider a fixed alphabet $\Sigma=\{a_1,\ldots,a_n\}$.

\begin{definition}\label{def_substitution}
Let $K_1,\ldots,K_n\subseteq\Gamma^*$ be $n$ languages on an alphabet $\Gamma$ that may be different from $\Sigma$. The $K_i$s  define a substitution $\rho:\Sigma^*\to\mathcal{P}(\Gamma^*)$ that associates a $\Gamma$-language with every word of $\Sigma^*$. Formally $\rho(w)$ is given by
\begin{xalignat*}{2}
\rho(a_{i_1}a_{i_2}\cdots a_{i_m})&=K_{i_1}\cdot K_{i_2}\cdots K_{i_m}\:,
&
\rho(\epsilon)&=\{\epsilon\}\:.
\end{xalignat*}
This is then lifted to $\Sigma$-languages with
\[
\rho(L)=\bigcup_{w\in L}\rho(w)\:.
\]
\end{definition}
We sometimes write $L^{a_1 \leftarrow K_1,...,a_n \leftarrow K_n}$ instead of $\rho(L)$: this notation is convenient, e.g., when $K_i=\{a_i\}$ for most $i$'s. In such cases we may just write, e.g., $x^{a_3\leftarrow K_3}$ with the meaning that letter different from $a_3$ are unchanged by the substitution.

It is well known that if all the $K_i$ are regular then  $\rho(L)$ is regular when $L$ is (and we say that $\rho$ is  a regular substitution).

Actually, the substitution preserves more : if each $K_i + \epsilon$ is downward closed (for $i=1,\ldots,n$)
 then $\rho$ preserves downward closedness (and we say it is a downward closed substitution). By showing this, we can imply in this case the regularity as downward closed languages are regular. 

\begin{proposition}\label{subts_closed}

	If $K_1,...,K_n$ are downward closed then $\rho(L)$ is downward closed (hence also regular). 

\end{proposition}

\begin{proof}
Let $z \in \rho(L) $ then $z\in\rho(x)$ for some $x=a_{i_1}\cdots a_{i_m}\in L$
and $z$ can be factorized as $z=z_1\cdots z_m$ with $z_i\in K_{a_i}$. If now
$y\subword z$ then $y$ is some $y_1\cdots y_m$ with $y_i\subword z_i$ for
$i=1,\ldots,m$, so $y_i\in K_{a_i}$ since each $K_{a_i}$ is
downward closed. Finally $y\in\rho(x)\subseteq \rho(L)$.

Hence,  $\rho(L)$ is downward closed. 
\qed    
\end{proof}

\begin{remark}\label{commutation dc}

We observe that for any substitution where the $(K_i)$ are non-empty languages (not necessarily
downward closed) we have  $\dc\bigl(L^{a_1 \leftarrow K_1,...,a_n \leftarrow K_n}\bigr) = \dc(L)^{a_1 \leftarrow \dc(K_1),...,a_n \leftarrow \dc(K_n)}$. 

\end{remark}

The case of substitutions is well known and studied, and we will see in the next proposition that there is an exponential lower bound for the state complexity of substitutions even in the case of downward closed $L$ and $\rho$.

\begin{example}\label{exp_subst_easy}

Let $L = \{ \epsilon, a_1,...,a_n \} = \dc(\Sigma)$ and, for $i=1,\ldots,n$, $K_i = ( \Sigma\setminus \{a_i\} )^*$. Write $L'$ for $\rho(L) = \bigcup_{i} A_i$. Now a word in $L'$ cannot use all letters of $\Sigma$ so if $x$ is the prefix of some word $xy$ in $L'$ then $x$ can only be continued by some $y$ that does not use all the letters of $\Sigma \setminus \Sigma(x)$. We see that
 the quotient $L'/x$ is $(\Sigma \setminus \Sigma(x) )^*$ so
there is a bijection between the subsets of $\Sigma$ and the quotients of $L'$. 	
Hence, $\kappa(\rho(L)) = 2^n$. 

\end{example}

\begin{proposition}\label{exp_lower_bound_subst}

	Let $\Sigma = \{ a_1,...,a_n\}$, and $A_i = (\Sigma \setminus \{a_i\})^*$. Then, let $A[i,j] = A_i \shuffle \dc(a_i^j)$, the words not having $a_i^{j+1}$ as subword. Let $L = \{ \epsilon, a_1,...,a_n\}$, $(m_i)_{1 \leq i \leq n} \in \mathbb{N}^n$. Let us define $K_i = A[i,m_i]$, then $\kappa(L^{a_1 \leftarrow K_1,...,a_n \leftarrow K_n}) = \prod_{1 \leq i \leq n}(m_i+2) =\prod_{1 \leq i \leq n} \kappa(K_i) $. 

\end{proposition} 

\begin{proof}

Let $u \in \Sigma^*$, $L/u = \{ \sum A[i,{m_i-j_i}] \ | \ a_i^{j_i} \mbox{ subword of u, } 1 \leq j_i \leq m_i+1 \}$ with the convention $A[i,j] = \emptyset$ if $j>i$. Indeed, if $u$ contains $j_i$ occurrences of $a_i$ then we can still have $m_i-j_i$ such occurrences in a suffix while staying in $K_i$, else we are not in $K_i$. Thus we can take $u$ such that it has $0 \leq j_i \leq m_i$ occurrences of $a_i$ and for each value in $\{ 0,...,m_i+1\}^n$ we get a different quotient of $K_i$, and we do this for each $a_i$. Thus, there are $\prod_{1 \leq i \leq n} (m_i+2)$ such quotients of $L^{a_1 \leftarrow K_1,...,a_n \leftarrow K_n}$. Finally, $K_i$ has $m_i+2$ quotients : $\{ A[i,j] \ | \ 0 \leq j \leq m_i +1 \} $. Hence $\kappa(\rho(L)) = \prod_{1 \leq i \leq n} \kappa(K_i) $.
\qed    
\end{proof} 

\begin{remark}

	We can actually see that  \Cref{exp_subst_easy} is a particular case of the previous proof with $m_i = 0$ for all $i$.

\end{remark}

As the state complexity of the substitution is exponential in the general case, it can be interesting to know if there exists a restriction for which the state complexity is polynomial. We thus consider singular substitutions, that is to say substitutions $L^{a \leftarrow K}$.\\

Let us start with the easiest case, when there exists words $u,v$ such that $L,K = \dc(u),\dc(v)$.

\begin{proposition}
 Let $L=\dc (u)$  and $K=\dc (v)$ with $u$ and $v$ two words, and $a \in \Sigma$ be a letter. Then $L^{\left\{ a \leftarrow K \right\}}=\dc(u^{\left\{a\leftarrow v\right\}})$ and $\kappa(L^{\left\{a \leftarrow K\right\}})=|u| - |u|_{a}+|u|_{a} |v| +2 \leq \kappa(L)\kappa(K)$. 
\end{proposition}

\begin{proof}
  Using  \Cref{commutation dc} we get that $L^{a \leftarrow K} = \dc(u^{a \leftarrow v})$. We then apply \Cref{state_complexity_word} and conclude.
\qed    
\end{proof}

We now want to study the quotients of a more general substitution $L^{a \leftarrow K}$, for this we explain how to inductively compute  quotients of substitutions. As in \Cref{def_substitution}, we write $\rho(R)$  for $R^{a \leftarrow K}$.

\begin{lemma}\cite{browzo2010}

	Let $L$ be a language, we denote by $L^{\epsilon} = \begin{cases} 					\emptyset , & \text{if}\ \epsilon \notin L\:, \\
      			\{\epsilon\}, & \text{otherwise.}
    			\end{cases}$ \\

	Let us introduce the basic rules of the computation of quotients given in \cite{browzo2010} : 
	\begin{itemize}
		\item 
    			$b/a = \begin{cases} \emptyset , & \text{if}\ b \neq a\:, \\
      			\epsilon, & \text{otherwise.}
    			\end{cases}$  		
  				
		\item $ (L \cup K)/a = L/a \cup K/a $.
		
		\item $(L \cdot K)/a = (L/a).K \cup L^{\epsilon}.(K/a)$.
	\end{itemize}

\end{lemma}

\begin{remark}

In the following we are going to use downward closed languages and their quotients which are also downward closed. As $L^\epsilon = \{ \epsilon\}$ for any non-empty downward closed language $L$, the concatenation rule becomes in this case : 
\begin{equation}\label{eq_quotient}
(L \cdot K)/a = (L/a)\cdot K + (K/a) \;\;\;\text{if $L$ is non-empty.}
\end{equation}
This property is essential to understand the structure of the quotients. 
\end{remark}

\begin{remark}
One could think that the case $K = \emptyset$ might be pathological. 
However n our setting one can see
that, if $L$ is downward closed, then $L^{a\leftarrow\emptyset}=L^{a\leftarrow\epsilon}$.
\end{remark}

Now that we have those basic rules for computing quotients, we want to see how we can apply them in our case. For this we need to introduce the SRE formalism (for Simple Regular Expressions) and show how the rules for computing quotients behave on what we call products of atoms.

\begin{definition}[\cite{Abdulla2004}]

An atom is a (particular case of) regular expression $\alpha$ that is either a letter-atom $a+\epsilon$ with $a \in \Sigma$, or a star-atom $B^*$ with $B\subseteq \Sigma$.

A product of atoms (or product) $I$ is a finite concatenation of atoms : $I = \prod_{1\leq i \leq n} \alpha_i$ where $\alpha_i$ is an atom. It is a regular expression and the empty product denotes $\epsilon$.

An SRE $E$ is a finite sum of products : $E = \sum_{1 \leq j \leq m } I_j$. The empty sum denotes $\epsilon$. We denote by $\llbracket E \rrbracket $ the language described by $E$. The SREs form a subclass of regular expressions. 

\end{definition}

\begin{remark}
Atoms are included in products, that are included in SRE. Furthermore, 
we might  abuse  notations and write $E$ when what we mean is $\llbracket E \rrbracket $. For example by saying $w \in E$ when $w$ is a word.

\end{remark}

The following theorem linking the SREs and downward closed languages justifies this new representation.
Recall that, in our context, $L$ is \emph{directed} $\equivdef$ 
 $\forall x,y\in L:\exists z\in L:x\subword z\land y\subword z$.

\begin{theorem}\cite{Abdulla2004,wqo}
A language $L$ on a finite alphabet $\Sigma$ is downward closed if and only if it can be defined by an SRE.  
\\

Furthermore $L$ is directed if and only if it can be defined by a product.
\end{theorem}
A corollary is that any downward closed language is a finite union of directed languages.
\\

Now that we introduced the concept of SREs, we need to understand how quotienting behaves on them. For now we will focus on quotients of products of atoms. 

\begin{definition}
Let $I = \prod_{i\leq n} \alpha_i$ be a product of atoms. A suffix product of $I$ is any product $I' = \prod_{i_0 \leq i \leq n} \alpha_i$ for some $i_0$. 
\end{definition}

We are now going to use this notion to describe the structure of the quotients of a product.

\begin{lemma}\label{quotients of products}

Let $I = \prod_{i\leq n} \alpha_i$ be a product, then the quotients of $I$ are exactly the suffix products of $I$, with also the empty set if $\llbracket I \rrbracket \neq \Sigma^*$. 

\end{lemma}

\begin{proof}

Let $R$ be a quotient of $I$, there is a word $x$ such that $R = I/x$. Let us proceed by induction on the length of $x$. \\

If $|x| = 0$, $I/x = I$ and $I$ is a suffix product of itself. \\	

Let $x = x'b $ with $|x'| \geq 0$, and let $I' = I/x'$ be the suffix product we get by the induction hypothesis. If $I' = \emptyset$, then $I/x = I'/b = I' = \emptyset$ and it is valid. Now if $I'$ is non-empty, using the definition of suffix product, $I' = \prod_{j' \leq i} \alpha_i$. Let $\alpha_{i_0} = \min_{j' \leq i} ( b \in \alpha_i)$. If such an $\alpha_{i_0}$ does not exist, then $I/b = \emptyset$. Using the rules for the computation of quotient, we have: 

\begin{itemize}
\item 
    			$(a+\epsilon)/b = \begin{cases} \emptyset , & \text{if}\ b \neq a \\
      			\{ \epsilon \}, & \text{otherwise}
    			\end{cases}$ 
    			
\item $(B^*)/b = \begin{cases} \emptyset , & \text{if}\ b \notin B \\
      			B^*, & \text{otherwise}
    			\end{cases}$
\end{itemize}

Thus, if $\alpha_{i_0}$ exists, then $I'/b = \alpha_{i_0} \prod_{{i_0} < i} \alpha_i$ if $\alpha_{i_0}$ is a star-atom and $I'/b = \epsilon \cdot \prod_{{i_0} < i} \alpha_i$ if $\alpha_{i_0}$ is a letter atom . Thus $I'/b$ is a suffix product of $I'$, which makes it a suffix product of $I$. \\

Finally, if $\llbracket I \rrbracket \neq \Sigma^*$ then there exists $w \in \Sigma^*$ such that $w \notin \llbracket I \rrbracket$, thus $I/w = \emptyset$. And if $\llbracket I \rrbracket = \Sigma^*$ then for all $w \in \Sigma^*$, $I/w = \Sigma^* = I \neq \emptyset$.  
\qed    
\end{proof}

\begin{corollary}\label{two_quotients_for_ideals}

Let $I$ be a product of atoms, and $R_1,R_2$ two quotients of $I$. Then one of the quotient contains the other. 

\end{corollary}

Having this result for the quotients of products of atoms will make it easier for us to prove that quotients once we apply the substitution can be expressed in a compact form in the case of downward closed language. For now let us show this result for product for atoms.  

\begin{lemma}\label{quotients of quotient on ideals}

Let $I$ be a product of atoms, $K$ a downward closed language and $a,b \in \Sigma$, then $\rho(I)/b = \begin{cases} \rho(I/b) + (K/b) \cdot \rho(I/a)  & \text{if}\ a \neq b\:, \\
      			(K/a) \cdot \rho(I/a) & \text{otherwise.}
    			\end{cases}$

\end{lemma}

\begin{proof}

Let us proceed by induction on the length of $I$. 
If $I$ is the empty product then both terms give $\emptyset$ hence the equality. \\

If $I = \alpha \cdot I'$, then $\rho(\alpha \cdot I')/b = (\rho(\alpha)/b )\cdot \rho(I') + 
 \rho(\alpha)^{\epsilon} \cdot \rho(I')/b$. Now as $\alpha$ and $I'$ are downward closed and $\alpha \neq \emptyset$ we can use \Cref{eq_quotient}. Hence, $\rho(I)/b = (\rho(\alpha)/b )\cdot \rho(I') +  \rho(I')/b$. \\
 
 \begin{itemize}
 	
	\item If $\alpha = (a+\epsilon)$ then $\rho(\alpha)/b = K/b$. If $K/b \neq \emptyset$, then $\rho(I')/b \subseteq (\rho(\alpha)/b )\cdot \rho(I') = (K/b).\rho(I/a)$, thus $\rho(I)/b = (K/b)\cdot \rho(I/a)$. And as $I/b \subseteq I/a$, we get $\rho(I/b) \subseteq \rho(I/a) \subseteq (K/b)\cdot \rho(I/a)$. Hence, $\rho(I)/b = \rho(I/b) + (K/b)\cdot \rho(I/a)$.
	
	If, on the other hand $K/b = \emptyset$, then $\rho(I)/b = \rho(I')/b = \rho(I'/b) + (K/b)\cdot \rho(I'/a)$ by induction. And as $K/b = \emptyset$ we get $\rho(I)/b = \rho(I')/b = \rho(I'/b)$. Thus if $a \neq b$ implies  $I'/b = I/b$ we get $\rho(I'/b) = \rho(I/b) = \rho(I/b) + (K/b)\cdot \rho(I/a) $ as $I = (a +\epsilon)I'$ and $K/b = \emptyset$. On the other hand, if $a=b$, then as $K/a = \emptyset$, $\rho(I)/a = \emptyset = K/a \cdot \rho(I/a)$. \\

	\item If $\alpha = (c+ \epsilon)$ with $c \neq a$, $\rho(I)/b = (\rho(c+\epsilon)\cdot \rho(I')) /b = ((c+\epsilon)\cdot \rho(I'))/b$. 
	Thus if $c = b$, $\rho(I)/b = \rho(I') + \rho(I'/b) = \rho(I') = \rho(I/b)$. And when $c \neq b$ and $b \neq a$, $\rho(I)/b = \rho(I')/b = \rho(I'/b) + (K/b)\cdot \rho(I'/a)$ by induction. However as $I = (c+\epsilon)I'$, $\rho(I'/b) + (K/b)\cdot \rho(I'/a) = \rho(I/b) + (K/b)\cdot \rho(I/a) = \rho(I)/b$. Finally, if $c \neq b$ and $b = a$, $\rho(I)/a = \rho(I')/a = (K/a)\cdot \rho(I'/a) = (K/a)\cdot \rho(I/a) $. \\
	
	\item If $\alpha = B^*$ with $a \notin B$, then if $b \in B$ we have $B^*/b \neq \emptyset$ which implies $(K/b)\cdot \rho(I/a) \subseteq \rho(I')/b \subseteq (\rho(\alpha)/b )\cdot \rho(I') = (B^*)\cdot \rho(I') = \rho(I/b)$.  Thus $\rho(I)/b = \rho(I/b) + (K/b)\cdot \rho(I/a)$. 
	
	When $b \notin B$ and $b\neq a$, $\rho(I)/b = \rho(I')/b = \rho(I'/b) + (K/b)\cdot \rho(I'/a)$ by induction. And as neither $a$ or $b$ is in $B$ it gives $\rho(I'/b) + (K/b)\cdot \rho(I'/a) =\rho(I/b) + (K/b)\cdot \rho(I/a) = \rho(I)/b $. On the other hand, when $b \notin B$ and $b = a$, $\rho(I)/b = \rho(I')/b = (K/a)\cdot \rho(I'/a) = (K/a)\cdot \rho(I/a) $. \\
	
	\item Finally, if $\alpha = B^*$ with $a \in B$, then $\rho(\alpha) = (\Sigma(K) \cup (B \setminus \{ a \}))^* = B'^*$. If $b \in B'$ and $a \neq b$ then $\rho(I)/b = B'^* \rho(I')$. Now if $b \notin \Sigma(K)$ then $B'^* \rho(I') = \rho(I/b) = \rho(I/b)+(K/b)\cdot \rho(I/a)$. If $b \in \Sigma(K)$, $B'^* \rho(I') = \rho(I/a) = (K/b)\cdot \rho(I/a) = \rho(I/b) + (K/b)\cdot \rho(I/a)$. In the case where $b \in B'$ and $a \neq b$, this means that $a \in \Sigma(K)$. Thus $\rho(I)/a = B'^* \rho(I') = \rho(I/a) = K/a \cdot \rho(I/a)$. 
	
	Now if $b \notin B'$ and $b \neq a$, $\rho(I)/b = \rho(I')/b = \rho(I'/b) + (K/b)\cdot \rho(I'/a)$ by induction. However as $b \notin B$, $\rho(I'/b) = \rho(I/b)$. And as $b \notin B'$, we get $K/b = \emptyset$, thus $(K/b)\cdot\rho(I'/a) =(K/b)\cdot\rho(I/a) = \emptyset$. Hence $\rho(I)/b = \rho(I'/b) = \rho(I/b) = \rho(I/b) + (K/b)\cdot \rho(I/a)$. Finally, if $b \notin B'$ and $b = a$ we get $K/a = \emptyset$ since $a \notin \Sigma(K)$. Thus $\rho(I)/a = \emptyset = K/a \cdot \rho(I/a).$\\
	
	Thus, in the case where $a=b$, we get that $\rho(I)/a = K/a\cdot \rho(I/a)$. 
	And when $a \neq b$ we get $\rho(I)/b = \rho(I/b) + K/b \cdot \rho(I/a)$.
\qed    	
 \end{itemize}
\end{proof}

This result is interesting but as we do not want to limit ourselves to product of atoms, we need to extend it. And as the substitution and the quotient distribute over sums, we can easily extend this lemma to general downward closed languages.

\begin{proposition}\label{calcul_residu_subst_down}

	Let $L,K \subset \Sigma^*$ be languages with L,K downward closed and $a\neq b \in \Sigma$, we compute the quotients of $L^{a \leftarrow K}$ as follows : 
	
	\begin{itemize}
	
	\item $\rho(L) / \epsilon = \rho(L)$, 
	
	\item $\rho(L) / b = \rho(L/b) + K/b \cdot \rho(L/a)$,
	
	\item $\rho(L) /a = K/a \cdot \rho(L/a)$.
	
	\end{itemize}	 
\end{proposition}
\begin{proof}
Since $L$ is downward closed it can be written as a sum of product $L = \sum_j I_j$  , 

Let $b\neq a$ a letter.
\begin{align*}
    \rho(L)/b    &= \rho\bigl( \sum_j I_j \bigr) /b \\ 
                &= \sum_j \rho(I_j)/b \\
                &= \sum_j \bigl( \rho(I_j/b) + (K/b).\rho(I_j/a) \bigr)  \\
\shortintertext{using \Cref{quotients of quotient on ideals}}
                &= \sum_j \rho(I_j/b) + \sum_j (K/b).\rho(I_j/a) \\
                &= \rho(L/b) + (K/b).\rho(L/a)\:.
\end{align*}
A similar computation shows $\rho(L)/a = K/a \cdot \rho(L/a)$. 
\qed    
\end{proof}

Now that we know the structure of the quotients by a letter of the substitution in the case of downward closed languages, we can study the state complexity (i.e., quotients by words).

\begin{proposition}\label{Psi_well_defined}

Let $L,K$ be downward closed languages such that $\Sigma(L) \cap \Sigma(K) = \emptyset$, with $K \neq \emptyset$. For every $R$ quotient of $\rho(L)$, there exists a pair $(P,Q)$ where $Q \in \mathcal{R}(L)$ and $P \in \mathcal{R}(K)\cup \{ \{ \epsilon \} \}$, such that $R = P . \rho(Q) $. 

\end{proposition}

\begin{proof}

Let us write $\Psi$ for the function associating, with every quotient $R$ of $\rho(L)$, a pair $(P,Q)$  such that $R = P . \rho(Q)$: we want to show that this function is well defined. For a quotient $R$ of $\rho(L)$, there exists $x \in \Sigma^*$ such that $R = \rho(L) / x$. Let us prove the proposition by induction on the length of $x$. \\

If $|x| = 0$, $R = \rho(L) = \{ \epsilon \}.\rho(L)$, thus the pair $(L,\{ \epsilon \})$ fulfills the claim. 

If $x = x'b$ with $|x'| \geq 0$ and $b$ a letter, let us write $R' = \rho(L) / x'$, by induction hypothesis we have that $ R' = P' . \rho(Q')$. Using \Cref{calcul_residu_subst_down} we get that $ R = R'/b = (P' \cdot \rho(Q'))/b = P'/b \cdot \rho(Q') + \epsilon \cdot \rho(Q')/b$. However, if $P'/b \neq \emptyset$, as $\rho(Q')/b \subseteq \rho(Q')$ because $L$ is downward closed and $\epsilon \in P'/b$, then $\rho(Q')/b \subseteq P'/b \cdot \rho(Q')$. Thus, if $P'/b \neq \emptyset$, then $R = P'/b \cdot \rho(Q')$.

If $P'/b = \emptyset$ we have $R = \rho(Q')/b $. If $a=b$ then $R = K/a \cdot \rho(Q'/a)$ and this fulfills the claim. On the other hand, if $b \neq a$, $ R = \rho(Q'/b) + K/b . \rho(Q'/a) $. This is not of the expected form. However, using the hypothesis that the alphabets for $L$ and $K$ are disjoints, we see that only one of the terms is non-empty. In both cases we have quotients of $K$ or $\{ \epsilon\}$ and of $L$.
\qed    
\end{proof}

\begin{remark}\label{remark_plus}

	If $L$ and $K$ are on the same alphabet and are downward closed, then the function $\Psi$ is not well defined. For example, with $L = \dc \mathtt{ab} \ +\dc \mathtt{ba}$ and $K = \dc \mathtt{bbc}$, we have that $L^{a \leftarrow K}/b = \dc \mathtt{bcb} \ + \dc \mathtt{bbc}$. 

\end{remark}

\begin{remark}\label{remark_non_empty}

Let $L,K$ be downward closed languages, if $R$ quotient of $\rho(L)$ is non-empty, then for all pairs $P,Q \in (\mathcal{R}(K)\cup \{ \{ \epsilon \} \}) \times \mathcal{R}(L) $ such that $R = P \cdot \rho(Q)$ we have  $P\neq \emptyset$ and $Q \neq \emptyset$. 

\end{remark}

\begin{theorem}\label{upperbound_replace_diff}

	Let $L,K$ be downward closed languages based on disjoints alphabets. Then $\kappa\bigl(L^{a \leftarrow K}\bigr) \leq \kappa(L)\kappa(K)$. 
	
\end{theorem}

\begin{proof}

Using \Cref{Psi_well_defined} we know that the function $\Psi$ is well defined when $L,K$ are downward closed languages on disjoints alphabets. 

Let us show that $\Psi$ is injective. Take any two quotients $R \neq R'$ of $\rho(L)$, and assume $\Psi(R)= \Psi(R')= (P,Q)$. Then $R = P\cdot \rho(Q) = R'$ which is absurd. Thus $\Psi$ is injective. \\ 

The injectivity of $\Psi$ gives the upper bound $\kappa(\rho(L)) \leq |(\mathcal{R}(K)\cup \{ \{ \epsilon \} \}) \times \mathcal{R}(L)| \leq (\kappa(K)+1)(\kappa(L))$. However, if $L \neq \Sigma(L)^*$ and $K \neq \Sigma(K)^*$, then they both have an empty quotients and the \Cref{remark_non_empty} applies. Hence all quotients of $\rho(L)$ have their image by $\Psi$ in a set having $(\kappa(K))(\kappa(L)-1)+1$ elements. Thus $\kappa(L^{a \leftarrow K}) \leq (\kappa(L)-1)\kappa(K)+1$. \\

We now need to deal with the cases where we do not have $\emptyset$ as a quotient in one of the languages. Those cases are easier but the methods are different to reach the upper bounds. 

\begin{itemize}

\item If $L = \Sigma(L)^*$, $\rho(L) = ( \Sigma(L)\setminus \{a\} \cup \Sigma(K) )^*$, which is equal to ${\Sigma'}^*$ on the result alphabet  $\Sigma'=\Sigma(L)\setminus \{a\} \cup \Sigma(K)$. Thus, in this case $\kappa(L^{a \leftarrow K})=1 \leq \kappa(L)\kappa(K)$.

\item If $K = \Sigma(K)^*$, using \Cref{calcul_residu_subst_down} when $b \in \Sigma(K)$  gives $\rho(L)/b = K/b \cdot \rho(L/a) = K \cdot \rho(L/a)$, and if $c \notin \Sigma(K)$, $\rho(L)/bc = (K \cdot \rho(L/a))/c = \rho(L/a)/c = \rho(L/ac)$. Thus a word $w \in (\Sigma(L) \cup \Sigma(K))^*$ can be factorized as $w=k_1 \cdot l_1 \cdots k_n \cdot l_n$ with $k_i \in \Sigma(K)^*$ and $l_i \in \Sigma(L)^*$. 
This gives \fbox{$\rho(L)/w = \rho(L/(a l_1 \cdots a l_n))$} if $l_n \neq \epsilon$ and \fbox{$\rho(L)/w = K \cdot \rho(L/(a l_1 \cdots l_{n-1} a))$}.
Thus $\kappa(\rho(L))\leq \kappa(L) = \kappa(L)\kappa(K)$.

\end{itemize}

Thus, in all those cases, $\kappa(L^{a \leftarrow K}) \leq \kappa(L)\kappa(K)$.
\qed    
\end{proof}

\begin{remark}

One can find in \Cref{substitution_automaton} the construction of the automaton for $\rho(L)$ based on one for $L$ and one for $K$. This can give an insight, a more visual argument for the definition of $\Psi$. 

\end{remark}

The following corollary is the extension of \Cref{upperbound_replace_diff} to a regular substitution.	

\begin{corollary}

Let $L$ and $(K_{a_i})_{a_i \in \Sigma}$ downward closed languages such that all $|\Sigma|+1$ languages have pairwise disjoints alphabets, then we have that
\[ \kappa(L^{a_1 \leftarrow K_{a_1}, ... ,a_n \leftarrow K_{a_n} }) \leq \kappa(L) \prod_{1 \leq i \leq n} \kappa(K_{a_i}) \]

\end{corollary}

\begin{proof}

As the alphabets are pairwise disjoints, we can decompose the substitution as a composition of singular substitutions : $L^{a_1 \leftarrow K_{a_1}, ... ,a_n \leftarrow K_{a_n} } = (...( L^{a_1 \leftarrow K_{a_1} } )^{a_2 \leftarrow  K_{a_2}} ... )^{a_n \leftarrow  K_{a_n}} $. 
Using the iteration of the bound from \Cref{upperbound_replace_diff} we have that $\kappa(L^{a_1 \leftarrow K_{a_1}, ... ,a_n \leftarrow K_{a_n} }) \leq \kappa(L) \prod_{1 \leq i \leq n} \kappa(K_{a_i})$.  
\qed    
\end{proof}

This result gives an exponential bound for the substitution, which hints that even under the restrictive hypothesis of pairwise disjoint alphabets it seems hard to do better than exponential. We proved in \Cref{exp_lower_bound_subst} an example that reached this bound without this hypothesis, it could be interesting to try to reach it with the hypothesis on the alphabets for arbitrary $\kappa(K_i)$.\\

Furthermore, while researching on the singular substitution we could not find a family of downward closed languages $(L_i,K_i)_i$ whose substitutions $(\rho_i(L_i))$ had an exponentially growing state complexity. We thus make the following conjecture :

\begin{conjecture}\label{conjecture_finale}
    If $L$ and $K$ are downward closed languages, then $\kappa(L^{a\leftarrow K}) \leq \kappa(L)\kappa(K)$. 
\end{conjecture}

In order to prove \Cref{conjecture_finale} we need to avoid any assumption on the alphabets. Let us first prove that if $L$ is a product of atoms then the quadratic bound holds. 

\begin{lemma}\label{psi for products}
Let $I$ be a product and $K$ a downward closed language, then any quotient of $ I ^{a \leftarrow K}$ can be written as $P_K . \rho(P_I)$ where $P_K$ is either a quotient of $K$ or $\{ \epsilon \}$ and $P_I$ is a quotient of $ I  $.

\end{lemma}

\begin{proof}

Let $R$ be a quotient of $I^{a \leftarrow K}$, $R = I^{a \leftarrow K}	/x$. Let us proceed by induction on the length of $x$. \\

If $|x| = 0$, $R = I^{a \leftarrow K} = \{ \epsilon \} . \rho(I)$ which verifies the claim. \\ 

If $x = x'b$, we write $R' = I^{a \leftarrow K}/x' = P'_K.\rho(P'_I)$. 
We have that $R = R'/b = P'_K/b . \rho(P'_I) + \rho(P'_I)/b$ as $P'_K/b$ is either empty or downward closed thus containing $\epsilon$. If $P'_K/b  \neq \emptyset$, we have $R = P'_K/b . \rho(P'_I) $ as $\rho(P'_I)/b \subseteq \rho(P'_I)$ . \\

Now if $P'_K/b = \emptyset$, $R = \rho(P'_I)/b = \rho(P'_I/b) + K/b.\rho(P'_I/a)$. Now as showed in \Cref{quotients of products}, there are $I_1,I_2 $ suffix products of the product of $P'_I$ such that $P'_I/b = I_1$ and $P'_I/a = I_2$. And as showed in \Cref{two_quotients_for_ideals}, one of the quotient contains another. 
\begin{itemize}

\item If $I_1 \subseteq I_2$, then $\rho(I_1) \subseteq K/b.\rho(I_2)$, thus $R = K/b . \rho(P'_I/a)$. 

\item If $I_2 \subseteq I_1$ and $I_2 \neq I_1$, then we have that $I_2 = I_1/a $. In fact we have that $I_1 = \prod_{j_1 \leq i}  \alpha_i$ and $I_2 = \prod_{j_2 \leq i}  \alpha_i$ with $j_1 < j_2$. And the property is $\alpha_{j_2} = \min_i ( a \in \alpha_i ) $ thus as $j_1 < j_2$, we have that $\alpha_{j_2}  = \min_{j_1 \leq i} ( a \in \alpha_i ) $. Hence, $I_1/a = I_2$. Now $\rho(I_1)/b = \rho(I_1/b) + K/b.\rho(I_1/a) =  \rho(I_1/b) + K/b.\rho(I_2)$. Thus $K/b.\rho(I_2) \subseteq \rho(I_1)/b \subseteq \rho(I_1) $. Thus, $R = \rho(I_1) = \{ \epsilon \} . \rho(P'_I/b)$. 
\qed    
\end{itemize}
\end{proof}

\begin{theorem}\label{thm_productstate}

Let $I$ be a product of atoms and $K$ a downward closed language, then $\kappa(I^{a \leftarrow K}) \leq \kappa(K)\kappa(I) $.

\end{theorem}

\begin{proof}

If we assume that $I,K$ are not $\Sigma^*$,
using \Cref{psi for products} we have that $R = P_K . \rho(P_I) = \emptyset$ iff $P_K = \emptyset$ or $P_I = \emptyset$. Thus, $\kappa(I^{a \leftarrow K}) \leq \kappa(K)(\kappa(I)-1)+1 \leq \kappa(K)\kappa(I)$. \\ 

Let us now consider the case where $I$ or $K$ is $\Sigma^*$. 

\begin{itemize}

\item If $I = \Sigma(I)^*$ then $I^{a \leftarrow K} = (\Sigma(I)\setminus \{a\} \cup \Sigma(K))^*$. Hence, $\kappa(I^{a \leftarrow K}) = 1 \leq \kappa(I)\kappa(K)$.

\item If $K = \Sigma(K)^*$ then if $I/a = \emptyset$, $\kappa(I^{a \leftarrow K}) = \kappa(I)$. Otherwise we are just replacing the $a$ in the atoms by $\Sigma(K)$. Hence $\kappa(I^{a \leftarrow K}) \leq \kappa(I) = \kappa(I)\kappa(K)$. 
\qed    
\end{itemize}
\end{proof}

\subsection{Singular substitution when $K = \dc (\Sigma(L))$}

As showed in \Cref{remark_plus}, there exists singular substitutions having quotients that cannot be written as products of quotients. To understand their apparition, we can study what happens when $K = \dc (\Sigma(L))$.

\begin{lemma}\label{lemme_R_rho}
	
	Let $I = \alpha_0 I'$ be a product that does not contain atoms of the type $B^*$ with $a \in B$, and $\rho$ the substitution of $a$ by $K = \Sigma(I) \cup \epsilon$. 
	 
	Then $\mathcal{R}(I) = \{ \alpha_0 I' \} \cup \mathcal{R}(I')$.
	And $\mathcal{R}(\rho(I)) = \{ \rho(\alpha_0) \rho(I') \} \cup \mathcal{R}(\rho(I'))$. 

\end{lemma}

\begin{proof}

Let us write $I = \prod_{0 \leq i \leq n} \alpha_i$. Then $\mathcal{R}(I) = \{ \prod_{j \leq i \leq n} \alpha_i  | 0 \leq j \leq n  \}$ as the quotients of a product are its suffix products as proved in \Cref{quotients of products}. Thus, $\mathcal{R}(I) = \{ \alpha_0 I' \} \cup \mathcal{R}(I') $. \\

Let us abuse notations and denote by $\mathcal{R}(L)\cdot K $ the set $\{ P \cdot K \ | \ P \in \mathcal{R}(L) \} $. Let us now prove $\mathcal{R}(\rho(I)) =  \mathcal{R}(\rho(\alpha_0)) \rho(I') \cup \mathcal{R}(\rho(I'))$. Let us proceed by double inclusion. 

Let $R$ be a quotient in $\mathcal{R}(\rho(I))$, then $R = \rho(I)/x$ with $x \in \Sigma^*$. Let $x'$ be the longest prefix of $x$ in $\rho(\alpha_0)$, such that $x = x' y$. If $x \neq x'$, $R = \rho(I')/y \in \mathcal{R}(I')$. Otherwise $R = \rho(\alpha_0)/x \ \rho(I') \in \mathcal{R}(\rho(\alpha_0)) \rho(I')$. \\

For the other direction, let $R \in \mathcal{R}(\rho(\alpha_0)) \rho(I')$, we have $R = \rho(\alpha_0)/x \  \rho(I')$. Using the proof above, $R = \rho(I)/x \in \mathcal{R}(\rho(I))$. Then, as $\rho$ preserves  downward closedness, $\mathcal{R}(\rho(I')) \subseteq \mathcal{R}(\rho(I))$. \\ 

Now, let us prove that in our case, $\mathcal{R}(\rho(\alpha_0)) \rho(I') \cup \mathcal{R}(\rho(I')) = \rho(\alpha_0) \rho(I') \cup \mathcal{R}(\rho(I'))$. 

\begin{itemize}

\item If $\alpha_0 = (b+\epsilon)$ with $b \neq a $, then $\mathcal{R}(\rho(\alpha_0)) \rho(I') \cup \mathcal{R}(\rho(I')) = \{ \emptyset , \rho(I'), (b+\epsilon) \rho(I')\} \cup \mathcal{R}(\rho(I')) = \{ (b+\epsilon) \rho(I')\} \cup \mathcal{R}(\rho(I')) = \rho(\alpha_0) \rho(I') \cup \mathcal{R}(\rho(I'))$ as $\emptyset$ and $ \rho(I')$ are already in $\mathcal{R}(\rho(I'))$ if $\rho(I')\neq \Sigma^*$, which is verified as we do not have atoms of the form $B^*$ with $a \in B$. \\

\item If $\alpha_0 = (a+\epsilon)$. Then $\mathcal{R}(\rho(\alpha_0)) \rho(I') \cup \mathcal{R}(\rho(I')) = \{ \emptyset , \rho(I'), K\rho(I')\} \cup \mathcal{R}(\rho(I')) = \{ K \rho(I')\} \cup \mathcal{R}(\rho(I')) = \rho(\alpha_0) \rho(I') \cup \mathcal{R}(\rho(I'))$. 

\item If $\alpha_0 = B^*$ with $a \notin B $, then $\mathcal{R}(\rho(\alpha_0)) \rho(I') \cup \mathcal{R}(\rho(I')) = \{ \emptyset , B^* \rho(I')\} \cup \mathcal{R}(\rho(I')) = \{ B^* \rho(I')\} \cup \mathcal{R}(\rho(I')) = \rho(\alpha_0) \rho(I') \cup \mathcal{R}(\rho(I'))$. 
\qed    
\end{itemize}
\end{proof}

\begin{proposition}\label{prop_commmmmute}

Let $I$ be a product that does not have atoms $B^*$ where $a \in B$, and $K = \Sigma(I) \cup \{ \epsilon \} $, then $\mathcal{R}(\rho(I)) = \rho(\mathcal{R}(I))$, meaning that $\rho$ and $\mathcal{R}$ commute. 

\end{proposition}

\begin{proof}
As $K$ and $I$ are based on the same alphabet, let us write $\Sigma = \Sigma(I)=\Sigma(K)$. 
First let us show that the proposition is true for atoms $\alpha$. 
\begin{itemize}

\item Let $\alpha = (b+\epsilon)$, when $b \neq a$ or $\alpha = B^*$ when $a \notin B$, then as $\rho$ is the identity on those atoms, the claim ensues. To prove this, we can observe that if an atom does not contain the letter $a$ then its quotients will not contain it either, thus $\rho$ will be the identity on the quotients also.

\item Finally, when $\alpha = (a + \epsilon)$, $\mathcal{R}(\rho(a+\epsilon)) = \mathcal{R}(K) = \{ K,\epsilon,\emptyset \} $. 

And $\rho(\mathcal{R}(a+\epsilon)) = \rho(\{a+\epsilon, \epsilon , \emptyset \}) = \{ K,\epsilon,\emptyset \} $.

\end{itemize}
Hence, the commutation works on those atoms.  \\

Now let us consider  products and prove the proposition by induction on the number of atoms in a product : 

If $I$ is the empty product then $\rho(\mathcal{R}(I)) = \emptyset = \mathcal{R}(\rho(I))$.

If $I = \alpha I'$, then we have : 

\begin{align*}
\mathcal{R}(\rho(\alpha I')) &=  \rho(\alpha) \rho(I')  \cup \mathcal{R}(\rho(I')) & \mbox{ using \Cref{lemme_R_rho} } \\
							&= \rho(\alpha I' ) \cup \mathcal{R}(\rho(I')) &  \mbox{ using the definition of } \ \rho \\
							&= \rho(\alpha I' ) \cup \rho(\mathcal{R}(I')) & \mbox{ by induction hypothesis } \\ 
							&= \rho( \alpha I' \cup \mathcal{R}(I') ) & \mbox{ using the properties of } \ \rho \\
							&= \rho(\mathcal{R}(\alpha I)) & \mbox{ using \Cref{lemme_R_rho} \qed}
\end{align*}
\end{proof}

\begin{remark}

The equality $\rho(I/x) = \rho(I)/x$ is false in general, the equality from \Cref{prop_commmmmute} only happens when considering all the quotients together. For example take $I = \dc (\mathtt{abc})$. We have $\rho(I/b) = \dc(c)$ and $\rho(I)/b = \dc(bc)$. However 
\[
\mathcal{R}(\rho(I)) = \{ \dc (\mathtt{(a+b+c)bc}),\dc (\mathtt{bc}) , \dc (\mathtt{c}) , \epsilon, \emptyset \} = \{ \rho(\dc ( \mathtt{abc})), \dc(\mathtt{bc}) , \dc (\mathtt{c}) , \epsilon, \emptyset \} = \rho(\mathcal{R}(I)).\] 
\end{remark}

\begin{remark}	
Even if the commutation is not true in the general case, what we really are interested in is the cardinality. Hence, it would be interesting to know if the following inequality holds : let $L$ be a downward closed language, and $K = \Sigma(L) \cup \{ \epsilon \} $, then $|\mathcal{R}(\rho(L))| \leq |\rho(\mathcal{R}(L))|$. 

\end{remark}

In the case where $L$ is a product and $K = \dc(\Sigma(L))$ we are under the assumptions of \Cref{thm_productstate}, but we can show that we do not increase the state complexity by another way that follow what we did in this section.

\begin{proposition}

Let $w \in \Sigma^*$, then $\kappa\bigl(\dc(w)^{a \leftarrow \Sigma(w)}\bigr) \leq \kappa(\dc(w))$.

\end{proposition}

\begin{proof}

Let us show that every quotient $R$ of $\dc(w)^{a \leftarrow \Sigma(w)}$ can be written as $\rho(P)$ with $P$ quotient of $\dc(w)$.\\

Let $R$ be a quotient of $\dc(w)^{a \leftarrow \Sigma(w)}$, there exists $x \in \Sigma^*$ such that $R = \dc(w)^{a \leftarrow \Sigma(w)}/x$. Let us proceed by induction on the length of $x$. If $x = \epsilon$, $R = \dc(w)^{a \leftarrow \Sigma(w)} = \rho(\dc(w)) $.\\

If $x = x'b$ then let $R' = \dc(w)^{a \leftarrow \Sigma(w)}/x'$. We have that $R' = \rho(P')$ by induction. Thus $R = R'/b = \rho(P')/b$. Using \Cref{calcul_residu_subst_down}, we get that $R = \rho(P'/b) + \epsilon . \rho(P'/a)$. However as $P'$ is a quotient of $\dc(w)$ it can be written as $P' = \dc(v)$ with $v$ a suffix of $w$. Thus if $b$ appears before $a$ in $v$ then $\rho(P'/a) \subseteq \rho(P'/b)$, else $\rho(P'/b) \subseteq \rho(P'/a)$. Thus there exists $P$ a quotient of $\dc(w)$ such that $R = \rho(P)$. \\ 

Finally, we also have to consider the empty quotient, $\emptyset = \rho(\emptyset)$. \\

Hence $\kappa(\dc(w)^{a \leftarrow \Sigma(w)}) \leq \kappa(\dc(w)) $.
\qed    
\end{proof}

\section{Conclusion}

In this study we considered the state complexity of the $n$-th root and substitution operations when restricted to upward closed and downward closed languages. These questions have not yet (to the best of our knowledge) been addressed. This can be seen as trying to better understand these two subregular classes of languages. We found that in the case of upward closed languages, the root operations still have an exponential lower bound. We did not conclude for downward closed languages, but conjecture that if the language is made of subwords of a word, then the state complexity will be quadratic for the square root. \\

We also studied the state complexity of the substitution in the case of downward closed languages. We showed an exponential upper bound in the general case, and when restricting to singular substitutions, we proved a quadratic upper bound when the two languages are on disjoint alphabets and conjecture a quadratic bound in the general case of two downward closed languages. Furthermore, when the language we make the substitution on is directed, we showed a quadratic upper bound. It might be interesting in future studies to answer the conjecture on the singular substitution of downward closed languages and define/study the substitution for upward closed languages.\\ 

This works takes part in the more general goal of better understanding upward closed and downward closed languages
that appear in practical applications such as automatic verification based on well-quasi-orderings. In this way, it is important to understand better how these classes of languages are structured and how they behave to come up with better data structures allowing to manipulate them more efficiently. \\

I warmly thank Philippe Schnoebelen again for his advice, feed-backs, suggestions and help in general.

\bibliographystyle{splncs04} % LNCS BibTeX style
\bibliography{sample} % Your bibliography file without the .bib extension

\newpage
\appendix

\section{Missing proofs on root operators}

An extremely useful tool to obtain lower bounds on the state complexity of languages is the use of dividing sets. 

\begin{definition}

	Let $L \subseteq \Sigma^*$, then $\mathcal{F} = \{ x_1,...,x_n \}$ is a dividing set for $L$ if and only if for all $\ x_i \neq x_j$,  there exists $z_{i,j} \in \Sigma^*$ such that $\ x_iz_{i,j}\in L \ $ and $ \ x_jz_{i,j} \notin L $ or $\ x_iz_{i,j}\notin L \ $ and $ \ x_jz_{i,j} \in L $.

\end{definition}

The following fact explains why those dividing sets are useful when studying the state complexity of a language. 

\begin{fact}

Let $\mathcal{F}$ be a dividing set of $L$. Then any DFA recognizing $L$ must have at least $|\mathcal{F}|$ states.

\end{fact}

\begin{proposition}\label{quotient_diff_alph_sqrt}

Let $v,w$ be two words with different alphabets: $\Sigma(v)\neq\Sigma(w)$. Then $\sqrt{\upc(V_n)}/v \neq \sqrt{\upc(V_n)}/w$. 

\end{proposition}

\begin{proof}

    As $v$ and $w$ are based on different alphabets, there is a letter of $V_n$ that is either in $v$ and not in $w$ or in $w$ and not in $v$. Without loss of generality, let us assume the letter $a$ is in $v$ but not in $w$.

    Then let $x$ be $V_n$ with $a$ removed. 
    Thus, $V_n$ is a subword of $(vx)^k$ but not of $(wx)^k$. Thus, $vx \in \sqrt[k]{\upc (V_n)}$ but $wx \notin \sqrt[k]{\upc (V_n)}$. Hence, $\sqrt[k]{\upc (V_n)}/v \neq \sqrt[k]{\upc (V_n)}/w$ since one contains $x$ but not the other.
\qed        
\end{proof}

\begin{proposition}\label{better_dividing_set_sqrt_down}

	Let $\mathcal{F}_n$ be a dividing set of maximal size for $\sqrt{\upc(V_n)}$. 
	
	Then $\mathcal{F}_n \bigcup \mathcal{F}_n a_{n+1} \bigcup a_{n+1}(\mathcal{F}_{n-1} \setminus \{ V_{n-1} \}) a_n$ is a dividing set for $\sqrt{\upc(V_{n+1})}$.

\end{proposition}

\begin{proof}

Let $u,v$ in $\mathcal{F}_n \bigcup \mathcal{F}_n a_{n+1}$, 
    If they do not come from the same summand, then their alphabet is not the same as one uses $a_{n+1}$ while the
    other does not. Thus by \Cref{quotient_diff_alph_sqrt}, their quotients are different.\\ 

    If they come from the same summand, first let us assume they come from $\mathcal{F}_n$. Then there is a word $x \in \Sigma^*$ such that $ux \in L_n$ and $vx \notin L_n$. Thus $V_n$ is a subword of $uxux$ but not of $vxvx	$, hence $u_{n+1}$ is a subword of $uxa_{n+1}uxa_{n+1}$ and not of $vxa_{n+1}vxa_{n+1}$. Thus, $L_{n+1}/u \neq L_{n+1}/v$. \\

    If they both are from $\mathcal{F}_n a_{n+1}$. Let's write $u = u'a_{n+1}$ and $v = v'a_{n+1}$. There is $x'
    \in \Sigma^*$ such that $V_n$ is subword of $u'x'u'x'$ but not of $v'x'v'x'$ or the other way around. Thus $u_{n+1}$ is a subword of $u'x'a_{n+1}u'x'a_{n+1}$ but not of $v'x'a_{n+1}v'x'a_{n+1}$ as $V_n$ is not a subword of it and $V_n$ is a subword of $V_{n+1}$.\\

Let us consider $x\neq y \in a_{n+1}(\mathcal{F}_{n-1} \setminus \{ V_{n-1} \}) a_n$. Let us write $x$ and $y$ as $x = a_{n+1}x''a_n$ and $y = a_{n+1}y''a_n$. As $x'',y'' \in \mathcal{F}_{n-1}$, there exists $z \in \Sigma^*$ such that $x''zx''z \in \upc (V_{n-1})$ and $y''zy''z \notin \upc (V_{n-1})$. Thus $V_{n+1} \subword a_{n+1}x''a_nz a_n a_{n+1} a_{n+1}x''a_n z a_n a_{n+1}$ but $V_{n+1}$ is not a subword of $a_{n+1}y''a_n z a_n a_{n+1}a_{n+1}y''a_n za_n a_{n+1}$. Thus, $x$ and $y$ give different quotients. \\

If $x \in a_{n+1}(\mathcal{F}_{n-1} \setminus \{ V_{n-1} \}) a_n$ and $y \in \mathcal{F}_n $, as $x$ contains
$a_{n+1}$ but not $y$,  the proof of \Cref{prop-2-n-quotients} entails that their quotients are different. \\

If $x \in a_{n+1}(\mathcal{F}_{n-1} \setminus \{ V_{n-1} \}) a_n$ and $y \in \mathcal{F}_n a_{n+1}$, let us write $x = a_{n+1}x''a_n$ and $y = y'a_{n+1}$. 

As $x''a_n \in \mathcal{F}_n$, if $x''a_n \neq y'$, there exists $z \in \Sigma^*$ such that $x'zx'z \in \upc (V_n)$ and $y'zy'z \notin \upc (V_n)$ where $x' = x''a_n$, as $x'\neq y'$. Thus, $V_{n+1} \subword a_{n+1}x'za_{n+1} a_{n+1}x'za_{n+1} = (xza_{n+1})^2$ but $V_{n+1}$ is not a subword of $y'a_{n+1}za_{n+1}  y'a_{n+1}za_{n+1} = (yza_{n+1})^2$. Thus $x$ and $y$ give different quotients. \\

If $x' = y'$, then $x = a_{n+1}x''a_n$, the nice thing is that by taking $z = V_{n-1}$ we have $ V_n \subword x''a_n z x'' a_n = x' z x' $ which gives $V_{n+1} \subword x''a_n a_{n+1} z x'' a_n a_{n+1} = y z y $ but $V_{n+1}$ is not a subword of $xzx = a_{n+1} x''a_n  z a_{n+1} x'' a_n $, as long as $x'' \neq V_{n-1}$. Thus $x,y$ give different quotients. \\

Thus, $a_{n+1}(\mathcal{F}_{n-1}\setminus V_{n-1})a_n \bigcup \mathcal{F}_{n} \bigcup \mathcal{F}_{n}a_{n+1}$ is a dividing set for $\sqrt{\upc(V_{n+1})}$.
\qed    
\end{proof}

\begin{proposition}
	
	For $n \geq 1$,  $\kappa(\sqrt{\upc(V_{n})}) \geq (\frac{1}{\sqrt{2}}-\frac{1}{4})(\sqrt{2}+1)^n \approx 0.46 \times 2.41^n$.

\end{proposition}

\begin{proof}

Using \Cref{better_dividing_set_sqrt_down} we have that $a_{n+1}(\mathcal{F}_{n-1}\setminus V_{n-1})a_n \bigcup \mathcal{F}_{n} \bigcup \mathcal{F}_{n}a_{n+1}$ is a dividing set for $\sqrt{\upc(V_{n+1})}$. Thus for all $n \in \mathbb{N}_{>0}$ we can create dividing sets $\mathcal{F}_n$ of $\sqrt{\upc(V_{n})}$ verifying $|\mathcal{F}_{n+2}| \geq 2|\mathcal{F}_{n+1}| + |\mathcal{F}_{n}|-1$. \\

Let $v_n = |\mathcal{F}_n| - \frac{1}{2}$, we have the equation : $ v_{n+2} \geq 2v_{n+1} + v_n $. Thus $v_n \geq w_n$ where $w_0 = v_0=-\frac{1}{2}$, $w_1 = v_1=\frac{3}{2}$ and $ w_{n+2} = 2w_{n+1} + w_n  $. And using the formula for recurrent sequences of order $2$, we get $w_n = (\frac{1}{\sqrt{2}}-\frac{1}{4})(\sqrt{2}+1)^n - (\frac{1}{\sqrt{2}}+\frac{1}{4})(1 - \sqrt{2})^n$. Thus asymptotically, $w_n \sim (\frac{1}{\sqrt{2}}-\frac{1}{4})(\sqrt{2}+1)^n$ and  $|\mathcal{F}_n|\geq 0.46(2.41)^n$. In fact, by computing the first values and then using the asymptotic analysis, we get $|\mathcal{F}_n|\geq 0.46(2.41)^n$ for $n\geq 1$. 
\qed    
\end{proof}

\begin{definition}\label{cut_shuffle}
    For $w \in \Sigma^*$ we denote with $CS(w)$ the set $\bigcup_{w=tv} t \shuffle v$. This is the set of words obtained from $w$ by one ``cut and shuffle'' move. 
\end{definition}

\begin{proposition}\label{prop-minimal-sqrt-Ln}
    Let $V_n$ be a $n$-letter word of length $n$, $CS(V_n) = \{ \mbox{minimal words of } \sqrt{\upc (V_n)} \} $ and 
    $\upc (CS(V_n)) = \sqrt{\upc (V_n)}$.
\end{proposition}

\begin{proof}
    Let $u$ be a word that can be obtained by a cut and shuffle of $V_n$, $u \in t \shuffle v $ where $ V_n = tv$.

    Let us take $r = vt$, we have 
\[ 
u^2 \in ( t \shuffle v )( v \shuffle t ) \subseteq V_n \shuffle r \subseteq \upc(V_n)\:.
\] 

    Hence, $u \in \sqrt{\upc (V_n)}$, and as $\sqrt{\upc (V_n)}$ is upward closed,
    $\upc (CS(V_n)) \subseteq \sqrt{\upc (V_n)}$. \newline 

    Let $u \in \sqrt{\upc (V_n)}$ such that $|u| = n$, which implies that $u$ is minimal in $\sqrt{\upc (V_n)}$. We have $u^2 \in \upc (V_n)$. 
    Let us call $u_1$ the longest prefix of $V_n$ that is a subword of $u$. 
Then as $V_n$ us a subword of $u^2$, we have that the suffix $u_2$ such that $V_n = u_1u_2$ is a subword of $u$ too. Finally, as $|u| = |V_n| = |u_1| + |u_2|$ and $u_1,u_2$ have different letters, we have that $u \in u_1 \shuffle u_2$. 

Hence $u \in CS(V_n)$, and since any upward closed $L$ is $\bigcup_{u\text{ minimal in }L}\upc(u)$,
we get $ \upc (CS(V_n)) = \sqrt{\upc (V_n)}$.
\qed    
\end{proof}

\begin{proposition}\label{inversion_counting}
    Let $V_n=a_1a_2\cdots a_n$ be word of length $n$ where all letters are different. Then $|CS(V_n)| = 2^n -n$.
\end{proposition}

\begin{proof}
	It is quite easy to see that for $w = tv$ there are $\binom{|w|}{|t|}$ different words in $t \shuffle v$, but one of them is $w$. Thus $|CS(V_n)| = \sum_{0\leq |t| \leq n-1} \binom{n}{|t|} - (n-1) = 2^n -n$.
\qed    
\end{proof}

\section{Missing proofs on the substitution operator}

\begin{proposition}\label{substitution_automaton}

The substitution operator preserves regularity. 

\end{proposition}

\begin{proof}
We know this since we proved it preserves the closedness and upward closed/downward closed languages are regular languages, but this is an occasion to explain the classical construction of the automaton for the substitution, which might help to visualize better the situation.\\ 

Let us consider the classical construction for the substitution. Let us call $\mathcal{A}_L,\mathcal{A}_K$ the
automata for $L,K$. We modify $\mathcal{A}_L$ in the following way: For every transition $q \step{a} q'$ in
$\mathcal{A}_L$ that reads the
letter $a$, we insert a copy of $\mathcal{A}_K$ without final states, create an $\epsilon$-transition from $q$ to the
initial state of (that copy of) $\mathcal{A}_K$, and from each final state of  $\mathcal{A}_K$ we  add an
$\epsilon$-transitions to $q'$. Finally, we remove the original $q\step{a}q'$ transition,
determinize the resulting automaton, and obtain a DFA for the substitution $L^{a\leftarrow K}$. \\

In order to prove that the construction is correct, let us call $\mathcal{A}$ the automaton before the determinization, and let us show that $L(\mathcal{A}) =  L^{a \leftarrow K}$. Since the determinization process preserves the language, we would have the validity of the automaton. \\

Let $x \in L^{a \leftarrow K}$, there exists $x_L \in L$ such that  
$x_L = x_1 a x_2 ... x_{n-1} a x_n $ with $x_i \in \Sigma^*$ and $x = x_1 y1 x_2 ... x_{n-1} y_{n-1} x_n $ with $y_i \in K$. Let $q_0 , ... q_r $ be the path of $x_L$ in $\mathcal{A}_L$. Then by inserting the path for $y_i$ in $\mathcal{A}_K$ where there are transitions labeled $a$ we get a path for $x$ in $\mathcal{A}$ using the $\epsilon$-transitions. Thus $L^{a \leftarrow K} \subseteq L(\mathcal{A})$. \\

Let $x \in L(\mathcal{A})$ then by studying the path that $x$ follows in the automaton, we factorize $x$ under the form $x = x_1 y1 x_2 ... x_{n-1} y_{n-1} x_n $ with $x_i \in \Sigma^*$ and $y_i \in K$. Furthermore, the construction implies that $x_L = x_1 a x_2 ... x_{n-1} a x_n  \in L$ as we inserted $\mathcal{A}_K$ on edges labeled by $a$. Thus $L(\mathcal{A}) \subseteq L^{a \leftarrow K}$.  
\qed    
\end{proof}

\section{Missing proofs on iteration operator}

\begin{proposition}\label{borne_L_k}
In the case of $L$ upward closed, with state complexity $n$, the upper bound $k(n-1) + 1$ on the state complexity of $L^k$ is tight for $|\Sigma | \geq 1$.

\end{proposition}

\begin{proof}

First of all the bound $k(n-1) + 1$ corresponds to the concatenation upper bound applied $k$ times. Thus is it a valid upper bound for $L^k$. We now have to show that it is actually tight. \\

Let $L_n = \texttt{a}^{n-1} \texttt{a}^*$ be the language of words having at least $n-1$ \texttt{a}, this gives $L_n^k =\texttt{a}^{(n-1)k}\texttt{a}^*$. Let $w \in L_n$, and $x$ a word such that $w \subword x$, then $x$ has more \texttt{a}'s than $w$, implying $x \in L_n$. 
The quotients of $L_n$ are $\{ \texttt{a}^i\texttt{a}^* \ | \ 0 \leq i \leq n-1 \}$. Hence, $L_n$ has state complexity $n$. Finally, using the same idea, $L_n^k$ has state complexity $(n-1)k+1$. 
\qed    
\end{proof}

\end{document}